\documentclass[11pt]{article}

\usepackage[margin=1in]{geometry}
\usepackage{amsmath,amssymb,amsthm}
\usepackage[colorlinks=true,linkcolor=blue,citecolor=blue,urlcolor=blue]{hyperref}

\newcommand{\E}{\mathbb{E}}
\newcommand{\eps}{\varepsilon}
\newcommand{\adeg}{\widetilde{\deg}}      
\newcommand{\Ot}{\widetilde{O}}           
\newcommand{\Omt}{\widetilde{\Omega}}     
\newcommand{\Tht}{\widetilde{\Theta}}     
\newcommand{\bn}{\{0,1\}}
\newcommand{\str}{\{0,1,*\}}
\newcommand{\one}{\mathbf{1}}
\newcommand{\C}{\operatorname{C}}
\newcommand{\Cb}[1]{\C_{\overline{#1}}}    
\newcommand{\UC}{\mathrm{UC}}
\DeclareMathOperator{\polylog}{polylog}
\DeclareMathOperator{\AND}{AND}
\DeclareMathOperator{\OR}{OR}
\DeclareMathOperator{\Qq}{Q}
\newcommand{\mem}{\operatorname{mem}}
\newcommand{\s}{\operatorname{s}}
\newcommand{\bs}{\operatorname{bs}}

\theoremstyle{plain}
\newtheorem{theorem}{Theorem}[section]
\newtheorem{lemma}[theorem]{Lemma}
\newtheorem{corollary}[theorem]{Corollary}
\newtheorem{proposition}[theorem]{Proposition}

\theoremstyle{definition}
\newtheorem{definition}[theorem]{Definition}
\theoremstyle{remark}
\newtheorem{remark}[theorem]{Remark}

\title{An Optimal Separation Between Certificate Complexity\\
and Approximate Degree}
\author{Kaspars Balodis\thanks{Center for Quantum Computer Science, Faculty of Science and
  Technology, University of Latvia. \texttt{kaspars.balodis2@lu.lv}}}
\date{August 31, 2026}

\begin{document}
\maketitle

\begin{abstract}
We prove that certificate complexity can be quartically larger than approximate degree.  More
precisely, we construct a family of total Boolean functions $G$ with
\[
   \C(G)\;=\;\Omt\bigl(\adeg(G)^{4}\bigr),
\]
where $\C$ denotes certificate complexity and $\adeg$ denotes $1/3$-approximate degree.
This is optimal up to polylogarithmic factors, since every total Boolean function $f$ satisfies
$\C(f)\le O(\adeg(f)^4)$ by the classical block-sensitivity bounds of Nisan and
Nisan--Szegedy.
Thus the result closes the gap between these two measures and improves the
previously best known separation
$\C(f)=\Omt(\adeg(f)^3)$ by Balodis, Ben-David, G\"o\"os, Jain, and Kothari.

The construction starts from the partial function they used to quadratically separate
$0$-certificate complexity from unambiguous $1$-certificate complexity.
It already has the required certificate hardness, but its $0$-certificates
are unstructured, which blocks the derivation of a low-degree verifier.
We keep its $1$-condition and restrict the $0$-inputs to those certified by
a structured family whose validity admits a low-degree approximant,
while preserving the quadratic hardness.
The partial function with its low-degree verifier is then fed through the cheat-sheet framework
to yield the total function $G$ with the claimed separation.

The main technical ingredient is an approximate polynomial that verifies
the certificate in degree $\Ot(\sqrt n)$.
The verifier forms a low-degree count $W$ of the candidate
$1$-certificates that remain compatible with the asserted $0$-certificate,
and tests whether this count is zero.
Crucially, the construction ensures that $W$ never exceeds $\Ot(n)$,
instead of the $\Theta(n^2)$ candidate pairs it counts.
Since the degree of a zero-test is governed by the range of the count rather than by
its number of summands, this brings the verification down to degree $\Ot(\sqrt n)$.

\end{abstract}

\section{Introduction}\label{sec:intro}

A recurring theme in the study of Boolean functions is that the many natural ways of measuring
their complexity---decision tree depth $D$, certificate complexity $\C$, (exact) degree $\deg$,
approximate degree $\adeg$, sensitivity and block sensitivity, and randomized and quantum query
complexity, among others---are all polynomially related for total functions
\cite{Nis91,NS,BBCMW,BdW,Huang}.  A central program in the area is to determine the \emph{exact}
polynomial relationships, that is, the best possible exponents relating each pair of measures.
Many of these exponents are now known, while others remain stubbornly open.

This paper concerns two of these measures.  The \emph{certificate complexity} $\C(f)$ is the
maximum, over inputs $x$, of the minimum number of coordinates of $x$ whose values force
$f(x)$.  Its one-sided variants $\C_1$ and $\C_0$ are the nondeterministic and conondeterministic
query measures, and $\C(f)\le D(f)$.  The \emph{approximate degree} $\adeg(f)$ is the least degree of a real
polynomial that approximates $f$ pointwise to error $1/3$.  Approximate degree is one of the
most useful lower-bound techniques in the area---most prominently $\adeg(f)\le 2\,\Qq(f)$ for
bounded-error quantum query complexity \cite{BBCMW}---and a central object of study in its own
right.

The best known bounds relating each pair of measures are collected in a now-standard table of
separations, assembled by Aaronson, Ben-David, and Kothari \cite{ABK}
and Aaronson et al.\ \cite{ABKRT} and reproduced in updated form by Iyer et al.~\cite{Iyer}.
The cell of that table relating $\C$ and $\adeg$ was one of those whose exponent had not been
pinned down: before the present and concurrent work, it recorded a cubic lower bound against a
quartic upper bound.

\subsection{Background}\label{sec:background}

\paragraph{The quartic ceiling.}
It has been known since the 1990s that for a total Boolean function the gap
between certificate complexity and approximate degree is at most quartic.
Nisan proved $\C(f)\le \s(f)\,\bs(f)$
\cite{Nis91} (see also \cite{BdW}), where $\s$ is sensitivity and $\bs$ is block sensitivity;
since $\s(f)\le\bs(f)$ this gives $\C(f)\le\bs(f)^2$, and Nisan and Szegedy proved
$\bs(f)\le 6\,\adeg(f)^2$ \cite{NS}.  Therefore
\begin{equation}\label{eq:ceiling}
  \C(f)\ \le\ O\bigl(\adeg(f)^4\bigr),
\end{equation}
so any separation $\C(f)\ge\adeg(f)^{c-o(1)}$ has exponent at most $c=4$.

\paragraph{The puzzle framework and the cubic record.}
The $n$-bit $\AND$ function shows a quadratic separation between certificate complexity and
approximate degree---its certificate complexity is $n$ and its approximate degree is
$\Theta(\sqrt n)$.
The first example of a total Boolean function showing a better-than-quadratic separation
was given by Ben-David, G\"o\"os, Jain, and Kothari~\cite{BGJK},
who recast the task of separating $\C$ from other measures as three equivalent
combinatorial \emph{puzzles} and proved \cite[Thm.~2]{BGJK} that a solution of exponent $\alpha$
to any one of them yields solutions of the same exponent to all three, up to logarithmic factors.

The form we use is the following.
\begin{quote}
\textbf{Puzzle II \cite[\S1]{BGJK}.}
For $\alpha>1$, does there exist a partial function $f$ together with an
$x\in f^{-1}(*)$ such that both $\Cb{0}(f,x)$ and $\Cb{1}(f,x)$ are at least
$\C(f)^{\alpha-o(1)}$?
\end{quote}
Here $\Cb{b}(f,x)$ is the least number of bits of $x$ that must be fixed to rule out
output $b$; the formal definition is given in Section~\ref{sec:prelim}.
An exponent-$\alpha$ solution totalizes to a function $g$ with
$\C_0(g)\ge\UC_1(g)^{\alpha-o(1)}$, where $\UC_1$ is unambiguous $1$-certificate complexity;
this combines the transformation ``$\mathrm{II}\Rightarrow\mathrm{I}$'' \cite[\S5.1]{BBGJK} with
the cheat-sheet framework of Aaronson, Ben-David, and Kothari \cite{ABK}, and is recorded below
as Lemma~\ref{lem:totalization}.  If, in addition, certificate validity admits a low-degree
verifier, the cheat-sheet degree principle (\cite[\S6.1]{BBGJK}, \cite[Lem.~19]{ABK}) supplies
the second bound
{\renewcommand{\theHequation}{equation.star}%
\begin{equation}\label{eq:star}
  \adeg(g)\ \le\ \Ot\!\bigl(\sqrt{\UC_1(g)}\bigr) ,
  \tag{$\star$}
\end{equation}}%
and eliminating $\UC_1$ between the two gives $\C(g)\ge\adeg(g)^{2\alpha-o(1)}$.
For a partial function inspired by the board game Hex, they established both ingredients
with $\alpha=1.5$, obtaining the cubic record $\C(g)=\Omt(\adeg(g)^3)$ \cite[Cor.~5]{BGJK}, as
published in \cite{BBGJK}.

\paragraph{The optimal puzzle solution and the missing verifier.}
The exponent $\alpha=2$ is the largest one can hope for: by the equivalence of the puzzles,
$\alpha$ is capped at $2$ by the universal quadratic unambiguous-DNF-to-CNF conversion
\cite[Fact~1, Thm.~2]{BGJK}.  Ben-David et al.\ conjectured that such an optimal solution exists
and anticipated its
consequence: they wrote \cite[\S2.3]{BGJK} that an $\alpha=2$ solution ``would yield near-optimal
lower bounds'' for the separations above---in particular it would push the cubic $\C$-vs-$\adeg$
bound to the quartic ceiling \eqref{eq:ceiling}.

Balodis \cite{Bal} constructed such an $\alpha=2$ solution: a partial function $f$
with a $*$-input $x$ at which
$\Cb{0}(f,x),\Cb{1}(f,x)\ge \C(f)^{2-o(1)}$.
This improved a multitude of separations---among them a quadratic $\C$-vs-$\deg$
separation, a cubic $\C$-vs-$\s$ separation, and
optimal (up to logarithmic factors) solutions to the Alon--Saks--Seymour problem
and the clique-vs-independent-set communication problem.
The two works were combined into the joint paper of Balodis,
Ben-David, G\"o\"os, Jain, and Kothari \cite{BBGJK}.
The one predicted improvement that did \emph{not} follow was precisely the $\C$-vs-$\adeg$
separation, which remained cubic.  The obstruction lies in the cheat-sheet step: the degree
bound \eqref{eq:star} is not automatic, and its derivation requires the validity of a
certificate to be approximable in low degree.  The structured Hex certificates admit such a
verifier, but no degree-$\Ot(\sqrt n)$ verifier was known for the $0$-certificates of the optimal
$\alpha=2$ construction.  Accordingly, the joint paper calls the approximate-degree corollary
``the trickiest,'' says ``we do not know how to derive it from our quadratic solution,'' and
instead uses the ``more structured'' $\alpha=1.5$ Hex construction
\cite[\S2.2]{BBGJK}.  Removing this obstruction---and thereby realizing the quartic separation
that the $\alpha=2$ hardness was expected to give---is the contribution of this paper.

\subsection{Our result}\label{sec:result}

\begin{theorem}\label{thm:main}
For infinitely many $n$ there is a total Boolean function
$G_n:\bn^{N_n}\to\bn$, where $N_n=\Tht(n^3)$, such that
\[
  \C(G_n)\ =\ \Omega(n^2)
  \qquad\text{and}\qquad
  \adeg(G_n)\ =\ \Tht(\sqrt n) ,
\]
and hence
\[
  \C(G_n)\ =\ \Omt\bigl(\adeg(G_n)^{4}\bigr).
\]
\end{theorem}

\noindent
By \eqref{eq:ceiling}, exponent $4$ cannot be exceeded, so Theorem~\ref{thm:main} is optimal up
to polylogarithmic factors.

\subsection{Proof overview}\label{sec:overview}

\paragraph{A verifiable quadratic puzzle solution.}
We start from the exponent-$2$ partial function of \cite{Bal}.  It already has the required
hardness: its certificate complexity is $\Ot(n)$, while both co-certificate complexities at a
suitable $*$-input are $\Omega(n^2)$.  Its $0$-side, however, is defined only through the
existence of some short but otherwise arbitrary partial assignment, with no structure that makes
validity easy to approximate.  We retain the condition defining the $1$-inputs but restrict the $0$-inputs to those
admitting a certificate from an explicit, structured family.  The first part of the proof shows
that this restriction preserves the exponent-$2$ co-certificate hardness and that the resulting
certificates still have width $\Ot(n)$.

\paragraph{Counting instead of conjoining.}
The main technical task is to approximate the predicate that decides whether a proposed
$0$-certificate is valid.  A direct verifier amounts to a conjunction of $\Theta(n^2)$ checks,
for which generic $\AND$-composition gives only a degree-$\Ot(n)$ bound.  Instead we aggregate the
pairwise checks into one nonnegative integer-valued polynomial $W$: it counts the candidate
$1$-certificates that remain compatible with the proposed $0$-certificate.  Although $W$ is a
sum of $\Theta(n^2)$ terms, it has degree only $\polylog(n)$, since each term is low-degree and
summation does not increase degree.  More importantly, the combinatorial sparsity built into the
construction ensures that
\[
  0\ \le\ W\ \le\ \Ot(n)
\]
on every Boolean input, even when the proposed certificate is malformed.  We combine $W$ with
the remaining $O(n)$ local validity checks into a single deficiency polynomial $T$, still of
degree $\polylog(n)$ and range $\Ot(n)$, such that the certificate is valid exactly when $T=0$.
A univariate zero-test of degree $O(\sqrt{\operatorname{range}(T)})$ then gives a
degree-$\Ot(\sqrt n)$ approximant for certificate validity.  The point is that the degree is
governed by the \emph{range} of the count, not by its number of summands.  The corresponding
$1$-certificate verifier is simpler and has the same asymptotic degree.

\paragraph{Cheat-sheet totalization.}
Finally we apply the black-box totalization of Lemma~\ref{lem:totalization} to the resulting
partial function and its two low-degree verifiers.  The exponent-$2$ co-certificate hardness
yields a total function $G_n$ with $\C(G_n)=\Omega(n^2)$, while the verifier bounds yield
$\adeg(G_n)=\Ot(\sqrt n)$.  The universal quartic upper bound \eqref{eq:ceiling} supplies the
matching lower bound $\adeg(G_n)=\Omega(\sqrt n)$, completing the proof of
Theorem~\ref{thm:main}.

\subsection{Concurrent work}\label{sec:concurrent}

After completion of the proof, but while this paper was being prepared, we learned that
Pabbaraju~\cite{Pab} had independently obtained the same quartic separation,
$\C(f)=\Omt(\adeg(f)^4)$ \cite[Thm.~4]{Pab}, from a different function.  His starting point is a
new signed set-pair system, which yields unambiguous DNFs of width $O(n)$ with
$\C_0=\Omega(n^2)$ and \emph{no} logarithmic loss; its structure is such that certificate
validity becomes an $\AND$ of only $O(n)$ local checks, so the $\AND$-composition route of
\cite{BGJK,BBGJK} already delivers a degree-$\Ot(\sqrt n)$ verifier.  In other words, that work
removes the bottleneck described above by replacing the hard function with one whose certificates
are $\AND$-friendly, whereas we keep the construction of \cite{Bal} and dispense with the $\AND$.
Beyond the shared quartic separation, \cite{Pab}
obtains a range of further consequences that do not follow from the construction here, including
a constant-gadget lifting theorem, an optimal refutation of the Alon--Saks--Seymour conjecture
and an optimal $\Omega(\log^2 n)$ conondeterministic lower bound for clique-vs-independent-set,
the logarithm-free separations $\C(f)=\Omega(\deg(f)^2)$ and $\C(f)=\Omega(\s(f)^3)$, and a
sample-compression lower bound.

\subsection{Organization}\label{sec:org}

Section~\ref{sec:prelim} fixes notation, collects approximate-degree facts, and states the
black-box totalization used at the end of the paper.  Section~\ref{sec:construction} constructs
the partial function $f$ and proves its quadratic co-certificate hardness.
Section~\ref{sec:certifier} proves the main new ingredient: total witness verifiers
$\varphi_0,\varphi_1$ of approximate degree $\Ot(\sqrt n)$.
Section~\ref{sec:assembly} applies the black box to prove Theorem~\ref{thm:main}.  Proofs of the inherited
surviving-pairs and totalization tools are included in the appendices.

\section{Preliminaries}\label{sec:prelim}

\subsection{Complexity measures}\label{sec:measures}

We write $[m]:=\{1,\dots,m\}$, take all logarithms to base $2$, and use
$\Ot,\Omt,\Tht$ to suppress factors polylogarithmic in the asymptotic parameter.

We use the certificate framework of \cite{BGJK,Bal}.  A (possibly partial) Boolean function is a
map $f:\bn^n\to\str$; an input $x$ is a \emph{$0$-input}, \emph{$1$-input}, or \emph{$*$-input}
according to the value $f(x)$, and we write $f^{-1}(\Sigma):=\{x\in\bn^n:f(x)\in\Sigma\}$ for
$\Sigma\subseteq\str$.  The function is \emph{total} if $f(x)\in\bn$ for every $x$.

\paragraph{Partial assignments.}
A \emph{partial assignment} (or \emph{partial input}) is a string $\rho\in\str^n$.  Its
\emph{fixed coordinates} are the indices $i$ with $\rho_i\ne *$, and its \emph{size} $|\rho|$ is
their number.  A total input $x\in\bn^n$ is \emph{consistent with} $\rho$---equivalently, $\rho$
\emph{agrees with} $x$---if $\rho_i=x_i$ at every fixed coordinate $i$; such an $x$ is a
\emph{completion} of $\rho$.  (Consistency is symmetric: we say interchangeably that $\rho$ is
consistent with $x$ and that $x$ is consistent with $\rho$.)

\paragraph{Certificates.}
Let $\Sigma\subseteq\str$.  A partial assignment $\rho$ is a \emph{$\Sigma$-certificate} for an
input $x$ if $\rho$ is consistent with $x$ and \emph{every} completion $x'$ of $\rho$ satisfies
$f(x')\in\Sigma$.  The \emph{$\Sigma$-certificate complexity} of $x$ and of $f$ are
\[
  \C_\Sigma(f,x)\ :=\ \min\{|\rho|:\rho\text{ is a }\Sigma\text{-certificate for }x\}
  \quad (x\in f^{-1}(\Sigma)),
  \qquad
  \C_\Sigma(f)\ :=\ \max_{x\in f^{-1}(\Sigma)}\C_\Sigma(f,x).
\]
We use the convention $\max\varnothing=0$.
For $b\in\bn$ a \emph{$b$-certificate} is a $\{b\}$-certificate---a partial assignment consistent
with $x$ on which $f$ is constantly $b$ over all completions---and we abbreviate
$\C_b:=\C_{\{b\}}$.  The \emph{certificate complexity} is $\C(f):=\max\{\C_0(f),\C_1(f)\}$.

\paragraph{Co-certificates.}
For a $*$-input $x\in f^{-1}(*)$ we also use the two \emph{co-certificate} complexities
\[
  \Cb{0}(f,x)\ :=\ \C_{\{1,*\}}(f,x),
  \qquad
  \Cb{1}(f,x)\ :=\ \C_{\{0,*\}}(f,x),
\]
following the $\bar 0,\bar 1$ shorthand of \cite{BGJK,Bal} for the output sets $\{1,*\}$ and
$\{0,*\}$.  Thus a $\bar b$-certificate is a partial assignment consistent with $x$ that
\emph{rules out} the value $b$, in that no completion of it takes the value $b$; intuitively
$\Cb{b}(f,x)$ measures how hard it is to prove ``$f(x)\ne b$''.

\paragraph{Unambiguous certificates and degree.}
$\UC_1(f)$ is the least $k$ such that $f^{-1}(1)$ admits a width-$k$ \emph{unambiguous} DNF: a
family of $1$-certificates each of size $\le k$ such that every $1$-input is consistent with
\emph{exactly one} of them.  For $0\le\eps<1/2$, the
\emph{$\eps$-approximate degree} $\adeg_\eps(f)$ is the least
degree of a real polynomial $p$ with $|p(x)-f(x)|\le\eps$ for all $x\in f^{-1}(\bn)$; we set
$\adeg(f):=\adeg_{1/3}(f)$.  Thus no boundedness condition is imposed on $p$ at $*$-inputs.
The \emph{(exact) degree} $\deg(f)$ of a total $f$ is the degree of
its unique multilinear real representation.

\subsection{Standard approximate-degree facts}

The verifier of Section~\ref{sec:certifier} decides whether an integer-valued count vanishes.
We first record a univariate ``zero-test'' of the form needed below: a low-degree polynomial
that equals $1$ at $0$ and is small on $\{1,\dots,L\}$.

\begin{lemma}[Univariate zero-test]\label{lem:zerotest}
For every integer $L\ge1$ there is a univariate real polynomial $q_L$ of degree $O(\sqrt L)$ with
\[
  q_L(0)=1
  \qquad\text{and}\qquad
  |q_L(t)|\le\tfrac13\ \text{ for every integer } t\in\{1,\dots,L\}.
\]
\end{lemma}

\begin{proof}
For $L\le4$ the set $\{1,\dots,L\}$ has at most four points, and Lagrange interpolation gives a
polynomial of degree at most $4$ with $q_L(0)=1$ and
$q_L(t)=0$ for $t\in\{1,\dots,L\}$; so
assume $L\ge5$.  Let $T_d$ be the degree-$d$ Chebyshev polynomial of the first kind, characterized
by $T_d(\cos\theta)=\cos(d\theta)$; thus $|T_d(x)|\le1$ for $|x|\le1$, and for $x\ge1$
\[
  T_d(x)=\cosh\!\bigl(d\,\operatorname{arccosh}x\bigr)\ \ge\ \tfrac12\,e^{\,d\operatorname{arccosh}x}.
\]
The affine map $u(t):=\dfrac{2t-(L+1)}{L-1}$ sends $t=1\mapsto-1$ and $t=L\mapsto1$, so
$u([1,L])=[-1,1]$, while $u(0)=-\bigl(1+\tfrac{2}{L-1}\bigr)$.  Put $\gamma:=\tfrac{2}{L-1}\le\tfrac12$
and
\[
  q_L(t)\ :=\ \frac{T_d\bigl(u(t)\bigr)}{T_d\bigl(u(0)\bigr)},
  \qquad d:=\bigl\lceil\,\ln 6\cdot\sqrt{L-1}\,\bigr\rceil=O(\sqrt L).
\]
Then $q_L(0)=1$ and $\deg q_L=d=O(\sqrt L)$.  For $t\in[1,L]$ we have $u(t)\in[-1,1]$, hence
$|T_d(u(t))|\le1$, and since $T_d$ has a definite parity $|T_d(u(0))|=T_d(1+\gamma)$; thus
\[
  |q_L(t)|\ \le\ \frac{1}{T_d(1+\gamma)}.
\]
Finally, we use the elementary estimate
\[
  \operatorname{arccosh}(1+\gamma)\ge \tfrac12\sqrt{2\gamma}\qquad(0<\gamma\le\tfrac12),
\]
which follows from
$\operatorname{arccosh}(1+\gamma)=\ln\!\bigl(1+\gamma+\sqrt{2\gamma+\gamma^2}\bigr)\ge
\ln\bigl(1+\sqrt{2\gamma}\bigr)\ge\sqrt{2\gamma}-\gamma\ge\tfrac12\sqrt{2\gamma}$
for this range of $\gamma$.  Together with $\tfrac12\sqrt{2\gamma}=1/\sqrt{L-1}$, this gives
$d\cdot\operatorname{arccosh}(1+\gamma)\ge d/\sqrt{L-1}\ge\ln 6$, whence
$T_d(1+\gamma)\ge\tfrac12 e^{\ln 6}=3$.  Therefore $|q_L(t)|\le\tfrac13$ for all $t\in[1,L]$.
\end{proof}

\noindent
We use the zero-test through the following immediate consequence, which makes precise the idea of
testing whether a low-degree integer-valued ``count'' is zero.

\begin{corollary}\label{cor:counting}
Let $W$ be a polynomial in Boolean variables whose multilinear representative has degree $d_W$
and whose Boolean-cube values lie in $\{0,1,\dots,L\}$.  Then $q_L\circ W$ (using the
multilinear representative of $W$) is a polynomial of degree $O(\sqrt L)\cdot d_W$ that takes
values in $[-\tfrac13,1]$ on Boolean inputs and
$\tfrac13$-approximates the predicate $[\,W=0\,]$.
\end{corollary}

\begin{proof}
The degree is at most $\deg(q_L)\cdot d_W=O(\sqrt L)\,d_W$.  On a Boolean input $W$ takes some
value $t\in\{0,\dots,L\}$; by Lemma~\ref{lem:zerotest}, $q_L(t)=1$ if $t=0$ and
$|q_L(t)|\le\tfrac13$ if $t\ge1$.  Hence $q_L\circ W$ lies in $[-\tfrac13,1]$ and equals the
indicator $[\,W=0\,]$ up to error $\tfrac13$.
\end{proof}

\begin{lemma}[Symmetric approximation; \cite{NS,Pat}]\label{lem:and}
For every integer $m\ge1$,
$\adeg(\AND_m)=\Theta(\sqrt m)$ and $\adeg(\OR_m)=\Theta(\sqrt m)$.
\end{lemma}

\begin{proof}
For the upper bound, $\AND_m(x)=1$ iff $W:=\sum_{i=1}^m(1-x_i)=0$, where $W$ is multilinear of
degree $1$ and integer-valued in $\{0,\dots,m\}$; Corollary~\ref{cor:counting} then yields an
approximant of degree $O(\sqrt m)$, and $\OR_m$ is dual.  The matching lower bound
$\Omega(\sqrt m)$ is due to Nisan--Szegedy and Paturi \cite{NS,Pat}.
\end{proof}

\begin{lemma}[Robust composition and error reduction; \cite{BNRW,Sherstov}]\label{lem:robust}
\leavevmode
\begin{enumerate}
  \item (Error reduction.)  Suppose $p$ $1/3$-approximates a Boolean function $g$ and
  $p(\bn^n)\subseteq[-\tfrac13,\tfrac43]$.  For every $0<\delta\le1/3$ there is a
  polynomial $p'$ of degree $O(\deg(p)\cdot\log(1/\delta))$ with $p'(x)\in[0,1]$ that
  $\delta$-approximates $g$.
  \item (Bounded \textup{\textsc{And}}.)  If $m\ge1$ and $g_1,\dots,g_m$ are total Boolean
  functions with $\adeg(g_i)\le d$, then
  $\adeg(g_1\wedge\cdots\wedge g_m)
  =O(d\cdot\sqrt m\cdot\log(m+1))$.
  \item (\textup{\textsc{And}} of exact inner functions.)  If $m\ge1$ and
  $h_1,\dots,h_m:\bn^n\to\bn$ are computed
  \emph{exactly} by polynomials of degree $\le d$, then
  $\adeg(h_1\wedge\cdots\wedge h_m)=O(\sqrt m\cdot d)$, by composing the degree-$O(\sqrt m)$
  approximant of $\AND_m$ from Lemma~\ref{lem:and} with the $h_i$.
\end{enumerate}
\end{lemma}

\subsection{A black-box totalization}\label{sec:blackbox}

We package the two standard cheat-sheet ingredients used later into one statement.  It combines
the totalization ``$\mathrm{II}\Rightarrow\mathrm{I}$'' of \cite[\S5.1,
Claims~16--17]{BBGJK} with the degree principle of \cite[\S6.1]{BBGJK} and
\cite[Lem.~19]{ABK}.  The persistence hypothesis below says that an accepted descriptor
identifies an ordinary certificate whose locations depend only on the descriptor.

\begin{lemma}[Verifiable-puzzle totalization]\label{lem:totalization}
Let $F:\bn^m\to\str$ be a partial function and $z\in F^{-1}(*)$.  For each $b\in\bn$,
suppose there is a total predicate $\varphi_b:\bn^m\times\bn^{\ell'}\to\bn$ such that
\[
  F(u)=b\quad\Longleftrightarrow\quad\exists y\ \varphi_b(u,y)=1.
\]
Suppose also that each descriptor $y$ determines a coordinate set
$I_b(y)\subseteq[m]$ of size at most $w$ such that, whenever $\varphi_b(u,y)=1$,
\[
  \varphi_b(u',y)=1
  \quad\text{for every $u'$ agreeing with $u$ on $I_b(y)$}.
\]
Fix $k\ge1$, put $H:=2^k$, and assume $H\le\mathrm{poly}(m)$.  Then there is a total
function $g$ on exactly $km+Hk\ell'$ bits with
\[
  \C_0(g)\ \ge\ \min\{H,\Cb{0}(F,z),\Cb{1}(F,z)\},
  \qquad
  \UC_1(g)\ \le\ k(\ell'+w).
\]
Moreover, if $\adeg(\varphi_b)\le d$ for $b\in\bn$, then
\[
  \adeg(g)\ \le\ \Ot(d).
\]
Concretely, the input to $g$ consists of $k$ base strings
$u^{(1)},\dots,u^{(k)}\in\bn^m$ and an array of $H$ cells, each containing $k$
descriptors.  If
$s:=(F(u^{(1)}),\dots,F(u^{(k)}))\in\bn^k$, then $g=1$ iff cell $s$ contains
descriptors accepted by the corresponding predicates; otherwise $g=0$.
\end{lemma}

A proof is included in Appendix~\ref{app:totalization} for completeness.

\section{A verifiable quadratic puzzle solution}\label{sec:construction}

This section constructs the hard partial function; Section~\ref{sec:certifier} supplies its
low-degree verifiers.  The complete package used in the final totalization is the following.

\begin{theorem}[Base function]\label{thm:base}
For every sufficiently large power of two $n$, there are a partial function
$f:\bn^{2n^2}\to\str$, an input $z\in f^{-1}(*)$, and total witness predicates
$\varphi_0,\varphi_1$ with a common descriptor length
$\ell'=\Theta(n\log n)$ such that:
\begin{enumerate}
  \item $n\le\C(f)\le O(n\log n)$ and
  \[
    \Cb{1}(f,z)\ge\frac{n(n-1)}2,
    \qquad
    \Cb{0}(f,z)\ge(\tfrac12+o(1))n^2;
  \]
  \item $f(x)=b$ iff some descriptor $y$ satisfies $\varphi_b(x,y)=1$;
  \item each accepted descriptor persists after fixing a descriptor-dependent set of at most
  $O(n\log n)$ input coordinates; and
  \item $\adeg(\varphi_0),\adeg(\varphi_1)=\Ot(\sqrt n)$.
\end{enumerate}
\end{theorem}

\subsection{Combinatorial setup}\label{sec:prior-construction}

\paragraph{Parameter conventions.}
Throughout the construction, $n$ ranges over powers of two and
\[
  \ell:=\log n,\qquad M:=\bigl\lfloor n^{(\ell+1)/(\ell+2)}\bigr\rfloor,\qquad L:=\ell n .
\]
We fix a bijection between $\{0,1\}^{\ell}$ and $[n]$, so every $\ell$-bit string names a row
or column.  For all sufficiently large $n$, $\ell\ge5$ and $1\le M<n$.  Writing
$M_0:=n^{(\ell+1)/(\ell+2)}$, we have $M\le M_0<M+1$ and, because $n=2^\ell$,
\[
  M=(\tfrac12+o(1))n,\qquad nM=(\tfrac12+o(1))n^2 ,
\]
where $o(1)$ is as $n\to\infty$.  These conventions make every verifier below a total
function of its Boolean descriptor bits.

We recall the part of the construction in \cite{Bal} used here.\footnote{The joint paper
\cite{BBGJK} recasts the optimal $\alpha=2$ solution in terms of \emph{everywhere almost-hittable}
hypergraphs; we instead build on the original matrix/matching formulation of \cite{Bal}, whose
association maps $r_1,\dots,r_\ell$ are exactly what the verifier of Section~\ref{sec:certifier}
exploits.}  An input $x\in\bn^{2n^2}$ consists of variables
$x_{i,j,b}$ with $i,j\in[n]$, $b\in\{1,2\}$, read as an $n\times n$ matrix whose $(i,j)$ entry
is the pair $(x_{i,j,1},x_{i,j,2})$; row $i$ is denoted $x_i$.  Two entries
$(a_1,a_2),(b_1,b_2)$ \emph{match} if $(a_1\wedge b_1)\vee(a_2\wedge b_2)=1$; two distinct rows
match if their entries match in every column.  A row is \emph{bad} if some entry equals $(0,0)$.
Note that a $(0,0)$ entry matches no entry, so two matching rows are automatically non-bad.

The next lemma is the combinatorial core of the construction: for a collection of fixed
``association'' maps, no moderately small set of rows can contain too many ordered pairs all of
whose associated rows also lie in the set.  We call such pairs the \emph{surviving} pairs of the
set; the bound on their number both keeps $0$-certificates short
(Section~\ref{sec:hardness}) and caps the range of the verifier's count $W$
(Section~\ref{sec:certifier}).

\begin{lemma}[Surviving pairs; {\cite[Lem.~1]{Bal}}]\label{lem:sparsity}
Let $n,\ell\in\mathbb N$ with $\ell\ge5$, and set
$M_0:=n^{(\ell+1)/(\ell+2)}$ and $L:=\ell n$.  There is a fixed
collection of maps $r_1,\dots,r_\ell:[n]\times[n]\to[n]$ such that for \emph{every}
$S\subseteq[n]$ with $|S|\le M_0$,
\[
  \bigl|\{(i,j)\in S\times S:\ \forall k\in[\ell]\ r_k(i,j)\in S\}\bigr|\ \le\ L .
\]
\end{lemma}

The statement in \cite{Bal} fixes $S$ in advance and bounds the count with probability $1-o(1)$
over random maps; the union bound in its proof gives the uniform version above, and is reproduced
in Appendix~\ref{app:sparsity}.

\noindent
With $\ell=\log n$, we have $M=(\tfrac12+o(1))n$ and $L=n\log n=\Ot(n)$.
Since the count in Lemma~\ref{lem:sparsity} is over
\emph{all} $(i,j)\in S\times S$, it dominates the count restricted to $i\ne j$, which is therefore
also $\le L$ for every $S$ with $|S|\le M$.

Given these fixed maps, the rows
$x_{r_1(i_1,i_2)},\dots,x_{r_\ell(i_1,i_2)}$ are the rows \emph{associated} with the ordered pair
$(i_1,i_2)$.  We say that $(i_1,i_2)$ is a \emph{clean matching pair} on $x$ if
$i_1\ne i_2$, the rows $x_{i_1},x_{i_2}$ match, and none of the associated rows is bad.

Let $P(x)$ denote the existence of a clean matching pair.  The partial function $f'$ introduced
in \cite{Bal} is defined by $f'(x)=1$ iff $P(x)$ holds; $f'(x)=0$ iff there is a partial
assignment $\rho$, consistent with $x$ and of size at most $(2\ell+2)n$, no completion of which
satisfies $P$; and $f'(x)=*$ otherwise.

\subsection{Canonical certificates and the partial function}\label{sec:function}

The function $f$ keeps the clean-matching $1$-condition of $f'$ and replaces its $0$-side by
a structured family of certificates.  We define those certificates first, so that the
subsequent definition of $f$ is self-contained.

\subsubsection{Canonical \texorpdfstring{$0$}{0}-certificate descriptors}\label{sec:encoding}

\begin{definition}[Canonical $0$-certificate descriptor]\label{def:canon}
A \emph{canonical $0$-certificate descriptor} is a pair $(j,A)$ where
\begin{itemize}
  \item $j\in[n]$ is a \emph{kill column};
  \item $A=(a_1,\dots,a_M)\in[n]^M$ is a \emph{row list} (each $a_s$ is a
  $\log n$-bit name), thought of as describing a set of rows.
\end{itemize}
\end{definition}

\noindent
The descriptor carries no pair list: the two zeros that spoil each surviving pair of its row set
are read directly from $x$ at canonical positions.  Its encoding length is
\[
  \ell'=(1+M)\log n=\Theta(n\log n)=\Ot(n).
\]

\paragraph{Binary-search membership.}
If
$C=(c_1,\dots,c_T)$ is a list over a totally ordered domain $\mathcal D$ and $y\in\mathcal D$,
let $\mem_C(y)$ be the result of the following lower-bound search, run even when $C$ is not
sorted.  Maintain an interval $[\lambda,\mu)\subseteq\{1,\dots,T+1\}$, initially
$[1,T+1)$.  While $\lambda<\mu$, read the middle entry
$c_m$ with $m=\lfloor(\lambda+\mu)/2\rfloor$; if $c_m<y$ set $\lambda:=m+1$, and otherwise set
$\mu:=m$.  At the end, accept iff $\lambda\le T$ and the entry $c_\lambda$ equals $y$.
For the row list $A$ the ordered domain is $[n]$, and we write
$S_A:=\{i\in[n]:\mem_A(i)=1\}$.

\begin{proposition}\label{prop:membership}
For every list $A$ (sorted or not) and every $i$, $\mem_A(i)$ depends on the descriptor bits
through a decision tree of depth $O(\log^2 n)$ (so it is computed exactly by a multilinear
polynomial of degree $O(\log^2 n)$ in the descriptor bits).  Moreover $\mem_A(i)=1$ only if
$i$ equals some name actually appearing in $A$, and distinct accepted values occupy distinct
list positions; hence $|S_A|\le M$.  If $A$ is sorted, then $\mem_A$ accepts exactly the values
appearing in $A$.
\end{proposition}

\begin{proof}
The lower-bound search makes $O(\log T)$ comparison probes on a length-$T$ list.  For $A$,
each comparison reads one $\log n$-bit row name and compares it to the fixed query $i$, so the
decision-tree depth is $O(\log M\cdot\log n)=O(\log^2 n)$.  A decision tree is computed exactly
by a multilinear polynomial of degree at most its worst-case depth, so
$\deg(\mem_A(i))=O(\log^2 n)$.

Acceptance requires the final inspected list entry to equal $i$, so every accepted $i$ is a name
that appears in $A$.  Charge such an accepted value to the final position $\lambda$ returned by
the search.  Two distinct accepted values cannot be charged to the same position, since that
position contains a single row name.  Thus $|S_A|\le M$, regardless of whether $A$ is sorted or
contains duplicates.  If $A$ is sorted and contains $i$, the usual lower-bound invariant says
that the returned position is the first entry at least $i$, hence that entry equals $i$; sorted
lists therefore realize true membership in their set of appearing values.
\end{proof}

\noindent
The point of Proposition~\ref{prop:membership} is robustness: \emph{whatever} the list
contents, $S_A$ is an honest set of size at most $M$.  When we build a descriptor
(Section~\ref{sec:hardness}) we sort $A$, so $\mem_A$ realizes true membership; when we
upper-bound degree on adversarial cell contents (Section~\ref{sec:certifier}), only the
bounded-size guarantee is used.

\subsubsection{The validity predicate}\label{sec:validity}

Fix the maps $r_1,\dots,r_\ell$ of Lemma~\ref{lem:sparsity}.  The intended semantics are
as follows.  The kill-column condition $(K)$ says that every row not accepted by $A$ is killed by
having a $(0,0)$ entry in column $j$.  Thus only rows in $S_A$ remain alive, and the ordered pairs
whose two rows and all of whose associated rows lie in $S_A$ are precisely the surviving pairs of
$S_A$ in the sense of Lemma~\ref{lem:sparsity}.  Each surviving pair must be \emph{spoiled} by two
canonical zeros certifying that its rows do not match.  The count $W$ is the number of surviving
pairs not so spoiled.

For an ordered pair of distinct rows $(u,v)$ define the \emph{residual indicator}
\[
  E_A(u,v)\ :=\ \mem_A(u)\,\mem_A(v)\,\prod_{k=1}^{\ell}\mem_A\bigl(r_k(u,v)\bigr)\ \in\bn,
\]
which is $1$ exactly when $u,v$ and all $\ell$ associated rows lie in $S_A$---that is, exactly on
the surviving pairs of $S_A$.  Define the \emph{canonical non-matching witness at column $v$}
\[
  Z_{u,v}\ :=\ (1-x_{u,v,1})(1-x_{v,v,2})\ \in\bn .
\]
Finally define the \emph{residual deficiency count}
\[
  \boxed{\;W(A,x)\ :=\ \sum_{\substack{u,v\in[n]\\ u\ne v}}
        E_A(u,v)\,\bigl(1-Z_{u,v}\bigr)\;}
\]
and the two validity conditions
\begin{align}
  (K):&\quad \forall i\in[n]:\quad \mem_A(i)=1 \ \ \text{or}\ \ (x_{i,j,1},x_{i,j,2})=(0,0);
        \label{eq:K}\\
  (R):&\quad W(A,x)=0. \label{eq:R}
\end{align}
\begin{definition}[Validity]\label{def:valid}
The canonical $0$-certificate descriptor $(j,A)$ is \emph{valid on $x$} if both \eqref{eq:K} and
\eqref{eq:R} hold.
\end{definition}

\begin{remark}\label{rem:Z}
The witness $Z_{u,v}$ correctly certifies non-matching of rows $u,v$ \emph{at column $v$}.
Indeed, if $Z_{u,v}=1$ then $x_{u,v,1}=0$ and $x_{v,v,2}=0$, so at column $v$,
$(x_{u,v,1}\wedge x_{v,v,1})\vee(x_{u,v,2}\wedge x_{v,v,2})=(0)\vee(x_{u,v,2}\wedge 0)=0$, i.e.\
the entries $(u,v)$ and $(v,v)$ do not match; hence rows $u$ and $v$ do not match.  Both bits
$x_{u,v,1}$ and $x_{v,v,2}$ are read at \emph{canonical} positions determined by $(u,v)$, so no
extra certificate data is needed.
\end{remark}

\subsubsection{The partial function and its witness predicates}

\begin{definition}\label{def:f}
$f$ is the partial Boolean function on $\bn^{2n^2}$ given by
\begin{itemize}
  \item $f(x)=1$ iff there is a clean matching pair on $x$;
  \item $f(x)=0$ iff some canonical $0$-certificate descriptor is valid on $x$;
  \item $f(x)=*$ otherwise.
\end{itemize}
The first two conditions are shown to be mutually exclusive below.
\end{definition}

The corresponding total witness predicates are
\[
  \varphi_0\bigl(x,(j,A)\bigr):=[\,(K)\wedge(R)\,],
  \qquad
  \varphi_1\bigl(x,(i_1,i_2)\bigr)
  :=[\,(i_1,i_2)\text{ is a clean matching pair on }x\,].
\]
The $0$-descriptor has length $(1+M)\log n=\Theta(n\log n)$; the $1$-descriptor has length
$2\log n$ and will be padded to the same length when Lemma~\ref{lem:totalization} is applied.

\subsubsection{Soundness and completeness}

\begin{lemma}[Soundness and completeness]\label{lem:A}
If some canonical $0$-certificate descriptor is valid on $x$, then there is no clean matching
pair on $x$.  Consequently $f$ is well defined and, for each $b\in\bn$,
\[
  f(x)=b\quad\Longleftrightarrow\quad
  \exists y\ \varphi_b(x,y)=1.
\]
\end{lemma}

\begin{proof}
Suppose $(j,A)$ is valid and $(u,v)$ is a clean matching pair.  Neither endpoint nor any
associated row is bad.  By $(K)$, all these rows therefore belong to $S_A$, so
$E_A(u,v)=1$.  Since $(R)$ says that the sum of the nonnegative Boolean terms
$E_A(a,b)(1-Z_{a,b})$ is zero, its $(u,v)$ term vanishes and $Z_{u,v}=1$.
Remark~\ref{rem:Z} then says that $x_u,x_v$ do not match, a contradiction.
The asserted witness characterizations now follow directly from Definition~\ref{def:f}.
\end{proof}

\subsection{Certificate width and persistence}\label{sec:certificates}

\begin{lemma}[Certificate complexity]\label{lem:certificates}
With $\ell=\log n$,
\[
  n\ \le\ \C(f)\ \le\ \max\{2n+2L,\ 2n(\ell+2)\}\ =\ 2n(\ell+2)\ =\ O(n\log n).
\]
\end{lemma}

\begin{proof}
\emph{Upper bound.}  If $f(x)=1$, fix a clean matching pair $(i_1,i_2)$ and read
the two matching rows and their $\ell$ associated rows.  These $\ell+2$ rows remain non-bad and
the first two remain matching under every completion, so they form a $1$-certificate of size
$\le 2n(\ell+2)=\Ot(n)$.  For a $0$-input $x$, fix a valid descriptor $(j,A)$
and let $\rho$ read the following \emph{input} bits: for each row $i\notin S_A$, the two
bits of entry $(i,j)$ (which equal $(0,0)$ by \eqref{eq:K}); and for each surviving pair
$(u,v)$ of $S_A$, i.e.\ each pair with $E_A(u,v)=1$, the two bits $x_{u,v,1},x_{v,v,2}$ (which
equal $0$ because \eqref{eq:R} forces $Z_{u,v}=1$).  By Lemma~\ref{lem:sparsity} and
Proposition~\ref{prop:membership}, there are at most $L$ such surviving pairs, so this reads
$\le 2n+2L=\Ot(n)$ bits.  Any completion $x'$ of $\rho$ keeps the \emph{same} certificate
descriptor $(j,A)$ valid: condition \eqref{eq:K} holds because the read $(0,0)$ entries persist;
condition \eqref{eq:R} holds because $E_A$ depends only on $A$, and every
surviving pair ($E_A=1$) has its witness bits forced to $Z=1$ by $\rho$, so every summand of
$W(A,x')$ vanishes.  Hence $f(x')=0$, so $\rho$ is a $0$-certificate and
$\C_0(f)\le 2n+2L=2n(\ell+1)$, which is dominated by the $1$-certificate bound.

\emph{Lower bound.}  Let $\one^{2n^2}$ be the all-ones input.  Every two rows match
and no row is bad, so $f(\one^{2n^2})=1$.  Suppose a partial assignment $\rho$ consistent with
this input fixes fewer than $n$ bits.  Some column $j$ is untouched.  Complete $\rho$ by setting
every entry in column $j$ to $(0,0)$, and let $A$ consist of $M$ repetitions of one fixed row.
Then $(K)$ holds, while $|S_A|=1$ implies that no distinct pair has $E_A=1$, so $(R)$ holds.
Thus $\rho$ has a $0$-completion and cannot be a $1$-certificate.  Hence
$\C(f)\ge\C_1(f,\one^{2n^2})\ge n$.
\end{proof}

\begin{proposition}[Persistent canonical certificates]\label{prop:persistence}
Each descriptor $y$ determines a set $I_b(y)$ of at most
\[
  w:=\max\{2n+2L,\ 2n(\ell+2)\}=2n(\ell+2)=O(n\log n)
\]
input coordinates such that, whenever $\varphi_b(x,y)=1$, the same descriptor is accepted on
every input agreeing with $x$ on $I_b(y)$.  In particular, those coordinates form an ordinary
$b$-certificate.
\end{proposition}

\begin{proof}
For $y=(j,A)$, take the two bits in column $j$ for every $i\notin S_A$ and the two witness bits
$x_{u,v,1},x_{v,v,2}$ for every pair with $E_A(u,v)=1$.  There are at most $2n+2L$ such
coordinates, and fixing them preserves $(K)$ and $(R)$.  For $y=(i_1,i_2)$, take all bits in
the two selected rows and their $\ell$ associated rows.  These at most $2n(\ell+2)$
coordinates preserve matching and non-badness.  Lemma~\ref{lem:A} then turns either persistent
witness into an ordinary certificate.
\end{proof}

\subsection{Quadratic hardness at the diagonal input}\label{sec:hardness}

The hard input is the \emph{diagonal input} $z\in\bn^{2n^2}$, defined by $z_{i,i}=(1,0)$ for all
$i\in[n]$ and $z_{i,j}=(0,1)$ for $i\ne j$.  Every entry of $z$ is $(1,0)$ or $(0,1)$---never
$(0,0)$---so $z$ has no bad row, and the diagonal entries keep any two distinct rows from matching.

\begin{lemma}[Quadratic co-certificate hardness]\label{lem:B}
The diagonal input satisfies $z\in f^{-1}(*)$ and
\[
  \Cb{1}(f,z)\ \ge\ \tfrac{n(n-1)}2,\qquad
  \Cb{0}(f,z)\ \ge\ nM\ =\ (\tfrac12+o(1))n^2 .
\]
In particular
$\min\{\Cb{0}(f,z),\Cb{1}(f,z)\}=\Omega(n^2)$.
\end{lemma}

\begin{proof}

\emph{$z$ is a $*$-input.}  As noted above, $z$ has no bad row and no two distinct rows match, so
no clean matching pair exists and $f(z)\ne1$.  Also $f(z)\ne0$: a valid canonical
$0$-descriptor requires, by \eqref{eq:K}, that every row $i\notin S_A$ be bad in column $j$; but
$z$ has no bad rows and $|S_A|\le M<n$, so some $i\notin S_A$ would have to be bad---impossible.
Hence $f(z)=*$.

\emph{$\Cb{1}(f,z)\ge n(n-1)/2$.}  Let $\rho$ be any partial assignment consistent with $z$
that rules out $f=1$.  For each unordered pair $\{a,b\}$, the rows $z_a,z_b$ can be
made matching by flipping $z_{a,b,1}$ and $z_{b,a,1}$ from $0$ to $1$; no bad row is created, so
unless $\rho$ fixes at least one of the two bits, this gives a $1$-completion of $\rho$.
Thus $\rho$ must read at least one of
$z_{a,b,1},z_{b,a,1}$ for every unordered pair $\{a,b\}$.  These obligations are disjoint, so
$\Cb{1}(f,z)\ge\binom n2=n(n-1)/2$.

\emph{$\Cb{0}(f,z)\ge nM$.}  Let
$\rho$ be consistent with $z$ with $|\rho|<h:=nM$.  We exhibit a
completion $x'$ of $\rho$ with $f(x')=0$, proving $\rho$ does not rule out $0$.  By averaging
over the $n$ columns, there is a column $j$ in which $\rho$ fixes fewer than $M$ bits (and hence
touches entries in at most $M$ rows); let $S\subseteq[n]$ be the rows whose
column-$j$ entry $\rho$ touches.  Form $x'$ from $\rho$ as follows:
\begin{itemize}
  \item For every row $i\notin S$, set $x'_{i,j}:=(0,0)$.  This is consistent with $\rho$
  because $\rho$ does not touch entry $(i,j)$ for $i\notin S$ (by definition of $S$).
  \item Let $E(S):=\{(u,v):u\ne v,\ u,v\in S,\ \forall k\ r_k(u,v)\in S\}$.  By
  Lemma~\ref{lem:sparsity}, $|E(S)|\le L$.  For each $(u,v)\in E(S)$ set
  $x'_{u,v,1}:=0$ and $x'_{v,v,2}:=0$.  This is consistent with $\rho$: since $\rho$ is
  consistent with $z$ and $z_{u,v,1}=0$ (as $u\ne v$) and $z_{v,v,2}=0$, the bits
  $x_{u,v,1},x_{v,v,2}$ are, wherever $\rho$ fixes them, fixed to $0$.
  \item Fix all remaining unset bits of $x'$ arbitrarily (consistently with $\rho$).
\end{itemize}
Now take the descriptor $(j,A)$, where $A$ is the sorted list of $S$, padded to length $M$ by
repeating its largest element when $S\ne\varnothing$---so that $A$ is again sorted---and is $M$
repetitions of row $1$ when $S=\varnothing$.  Since $A$ is sorted in both cases,
Proposition~\ref{prop:membership} gives $S_A=S$ unless $S=\varnothing$, in which case
$S_A=\{1\}$.
Condition \eqref{eq:K} holds because every
row outside the original set $S$ is bad in column $j$ by construction.  Condition \eqref{eq:R}
also holds: if $S=\varnothing$, then no distinct ordered pair has $E_A=1$; otherwise, a pair
$(u,v)$ with $E_A(u,v)=1$ lies in $E(S)$, and the witness bits give $Z_{u,v}=1$.
Hence every summand of $W$ is zero.  Therefore $(j,A)$ is valid on $x'$ and $f(x')=0$.
As $|\rho|<h$ was arbitrary,
$\Cb{0}(f,z)\ge h=nM=(\tfrac12+o(1))n^2$.
\end{proof}

\begin{remark}[Restricting to canonical descriptors is free]\label{rem:holeE}
A priori, replacing the $0$-side of $f'$ by only canonical descriptors could only
\emph{lower} $\Cb{0}(f,z)$: it makes $0$-inputs rarer and could make ``$\rho$ rules out $0$''
easier.  The reason it does not is that the certificate produced in the $\Cb{0}$ proof of
\cite[Claim~1]{Bal}
can be chosen in our canonical form: a kill column and a row set determine at most $L$
surviving pairs by Lemma~\ref{lem:sparsity}.  The killed entries are untouched by $\rho$, while
the canonically positioned spoiling bits are forced to zero by consistency with $z$.  Thus the
$\alpha=2$ hardness survives the restriction to this structured descriptor family.  Moreover,
the extracted $0$-certificate in Lemma~\ref{lem:certificates} has size at most
$2n+2L=(2\ell+2)n$ and rules out every clean matching pair by Lemma~\ref{lem:A}; hence
$f^{-1}(0)\subseteq(f')^{-1}(0)$.
\end{remark}

\section{Low-degree verification}\label{sec:certifier}

This section proves the main new ingredient.  Lemma~\ref{lem:totalization} reduces the
approximate degree of the final total function to that of the witness predicates defined in
Section~\ref{sec:function}:
\[
  \varphi_0\bigl(x,(j,A)\bigr)=[\,(K)\wedge(R)\,],
  \qquad
  \varphi_1\bigl(x,(i_1,i_2)\bigr)=[\,(i_1,i_2)\text{ is a clean matching pair}\,].
\]

\begin{lemma}[Low-degree verifiers]\label{lem:certifier}
$\adeg(\varphi_0)=\Ot(\sqrt n)$ and $\adeg(\varphi_1)=\Ot(\sqrt n)$.
\end{lemma}

The proof occupies the rest of the section.  Write $d_0:=c\log^2 n$ for the exact degree
bound of Proposition~\ref{prop:membership} (some constant $c$).

\subsection{The count \texorpdfstring{$W$}{W} is low-degree and bounded}

\begin{proposition}\label{prop:Wdeg}
After multilinearizing on the Boolean cube, $W$ is a polynomial in the input bits $x$ and the
descriptor bits of $A$ with
$\deg(W)=O(\ell\log^2 n)=O(\log^3 n)=\polylog(n)$.
\end{proposition}

\begin{proof}
Fix $u\ne v$.  Proposition~\ref{prop:membership} gives exact degree-$d_0$ polynomials for
$\mem_A(u)$, $\mem_A(v)$, and each $\mem_A(r_k(u,v))$.  Here every $r_k(u,v)$ is a fixed
index.  The product defining $E_A(u,v)$ has $\ell+2$ such factors and therefore degree at most
$(\ell+2)d_0$.
Since $Z_{u,v}=(1-x_{u,v,1})(1-x_{v,v,2})$ has degree $2$, the summand
$E_A(u,v)(1-Z_{u,v})$ has degree at most
$(\ell+2)d_0+2=O(\ell\log^2 n)$.  Summing the $n^2-n$ summands does not increase the
degree.  If a product creates repeated variables, replacing each $y^r$ by $y$ gives the
unique multilinear polynomial agreeing with it on Boolean inputs and cannot increase degree.
Therefore the multilinear representative of $W$ has degree $O(\ell\log^2 n)=O(\log^3 n)$.
\end{proof}

\begin{proposition}\label{prop:Wbound}
For \emph{every} Boolean assignment to $x$ and $A$---including a malformed row list---we have
$W(A,x)\in\{0,1,\dots,L\}$.
\end{proposition}

\begin{proof}
Each summand is a product of $\bn$-valued quantities, hence lies in $\{0,1\}$; so $W$ is a
non-negative integer and $W\le \#\{(u,v):u\ne v,\ E_A(u,v)=1\}$.  Now $E_A(u,v)=1$ forces
$u,v\in S_A$ and $r_k(u,v)\in S_A$ for all $k$.  By Proposition~\ref{prop:membership},
$|S_A|\le M$, so Lemma~\ref{lem:sparsity} (applied to the set $S=S_A$) gives
\[
  \#\{(u,v):u\ne v,\ E_A(u,v)=1\}\ \le\
  \bigl|\{(u,v)\in S_A\times S_A:\forall k\ r_k(u,v)\in S_A\}\bigr|\ \le\ L .
\]
Hence $0\le W\le L$.
\end{proof}

\noindent
Proposition~\ref{prop:Wbound} is the linchpin.  We do not need to test whether the row list is
sorted: binary-search membership guarantees a bounded accepted set even for an arbitrary,
unsorted descriptor.  Thus $W\le L$ holds on the entire Boolean cube, as required for the
univariate zero-test.

\subsection{A single zero-test for validity}

\begin{proposition}[The $0$-verifier]\label{prop:zerotest}
There is a single polynomial of degree $O(\sqrt{n+L}\,\log^3 n)=\Ot(\sqrt n)$ that
$\tfrac13$-approximates $\varphi_0$; hence $\adeg(\varphi_0)=\Ot(\sqrt n)$.
\end{proposition}

\begin{proof}
For each fixed $i\in[n]$, the local check
$\kappa_i:=\bigl[\mem_A(i)=1\bigr]\vee\bigl[(x_{i,j,1},x_{i,j,2})=(0,0)\bigr]$ is computed
\emph{exactly} by a decision tree of depth $O(\log^2 n)$: run the binary search for $i$ in $A$
($O(\log^2 n)$ probes of descriptor bits); separately read the $\log n$ bits of
the kill column $j$ and then the two input bits $x_{i,j,1},x_{i,j,2}$ ($O(\log n)$ probes);
output $\mem_A(i)\vee\overline{x_{i,j,1}}\,\overline{x_{i,j,2}}$.  Hence $\kappa_i$ has an exact
polynomial of degree $O(\log^2 n)$.  Define the total deficiency
\[
  T(j,A,x)\ :=\ \sum_{i=1}^n(1-\kappa_i)+W(A,x).
\]
On every Boolean input, $T$ is an integer in $\{0,\dots,n+L\}$ and has multilinear degree
$O(\log^3 n)$ by Proposition~\ref{prop:Wdeg}.  Because all its summands are nonnegative,
\[
  T=0\quad\Longleftrightarrow\quad (K)\wedge(R)
  \quad\Longleftrightarrow\quad \varphi_0=1.
\]
Corollary~\ref{cor:counting}, applied once to $T$, therefore gives a polynomial
$q_{n+L}\circ T$ that $1/3$-approximates $\varphi_0$ and has degree
$O(\sqrt{n+L}\,\log^3 n)$, which is $\Ot(\sqrt n)$ because $L=n\log n$.
\end{proof}

\noindent
Note that a single application of Corollary~\ref{cor:counting} suffices: the $n$ kill-column
checks and the $\Theta(n^2)$ pairwise conditions are merged into one statistic before any
approximation happens, so the degree is charged once, against the range $n+L=\Ot(n)$.

\begin{proof}[Proof of Lemma~\ref{lem:certifier}]
\emph{The $0$-verifier.}  This is Proposition~\ref{prop:zerotest}.

\emph{The $1$-verifier.}  $\varphi_1\bigl(x,(i_1,i_2)\bigr)$ accepts iff $i_1\ne i_2$, the rows
$x_{i_1},x_{i_2}$ match in every column, and no associated row $x_{r_k(i_1,i_2)}$ is bad.  This
is the conjunction of: one inequality check ($O(\log n)$ bits of the certificate); $n$
matching checks, the $j$-th reading the $2\log n$ bits naming $i_1,i_2$ and then
the four input bits $x_{i_1,j,\cdot},x_{i_2,j,\cdot}$ to evaluate
$(x_{i_1,j,1}\wedge x_{i_2,j,1})\vee(x_{i_1,j,2}\wedge x_{i_2,j,2})$; and $\ell n$ non-badness
checks, the $(k,j)$-th reading the name $r_k(i_1,i_2)$ (determined by the fixed map after
reading the $2\log n$ descriptor bits naming $i_1,i_2$) and the two input bits
of entry $(r_k(i_1,i_2),j)$ to evaluate $x_{r_k(i_1,i_2),j,1}\vee x_{r_k(i_1,i_2),j,2}$.  Each
of these $1+n+\ell n=\Ot(n)$ checks is computed exactly by a decision tree of depth
$O(\log n)$, so by Lemma~\ref{lem:robust}(3),
$\adeg(\varphi_1)=O(\sqrt{\ell n})\cdot O(\log n)=\Ot(\sqrt n)$.

Both verifiers have approximate degree $\Ot(\sqrt n)$, proving the lemma.
\end{proof}

\begin{proof}[Proof of Theorem~\ref{thm:base}]
The function and hard input are given by Definition~\ref{def:f} and the definition preceding
Lemma~\ref{lem:B}.  Lemmas~\ref{lem:certificates} and~\ref{lem:B} give item~1.
Lemma~\ref{lem:A} gives item~2, Proposition~\ref{prop:persistence} gives item~3, and
Lemma~\ref{lem:certifier} gives item~4.  Finally the longer descriptor has length
$(1+M)\log n=\Theta(n\log n)$, and padding the $1$-descriptor does not change its predicate
or approximate degree.
\end{proof}

\section{Proof of the main theorem}\label{sec:assembly}

\begin{proof}[Proof of Theorem~\ref{thm:main}]
Apply Lemma~\ref{lem:totalization} to the function, hard input, and witness predicates supplied
by Theorem~\ref{thm:base}.  Their common descriptor length is
$\ell'=(1+M)\log n=\Theta(n\log n)$, and their persistent certificate width satisfies
\[
  w\ \le\ \max\{2n+2L,\ 2n(\ell+2)\}=2n(\ell+2)=O(n\log n).
\]
The base arity is $m=2n^2$.  Choose
\[
  k:=2\log n,\qquad H:=2^k=n^2 .
\]
Let $G_n$ be the resulting total function.  Its input length is
\[
  N_n=km+Hk\ell'=\Theta(n^3\log^2 n)=\Ot(n^3).
\]

\smallskip
\emph{Certificate complexity lower bound.}  By Lemma~\ref{lem:totalization} and
Theorem~\ref{thm:base},
\[
  \C(G_n)\ \ge\ \C_0(G_n)\ \ge\
  \min\{H,\Cb{0}(f,z),\Cb{1}(f,z)\}
  =\Omega(n^2).
\]

\smallskip
\emph{Approximate degree upper bound.}  Theorem~\ref{thm:base} gives
$\adeg(\varphi_0),\adeg(\varphi_1)=\Ot(\sqrt n)$.  Since
$H=n^2=\mathrm{poly}(m)$, Lemma~\ref{lem:totalization} gives
\[
  \adeg(G_n)\ \le\ \Ot\!\left(\max_{b\in\bn}\adeg(\varphi_b)\right)
  \ =\ \Ot(\sqrt n).
\]

\smallskip
\emph{The approximate degree is $\Tht(\sqrt n)$.}  The matching lower bound is free: by the
universal ceiling \eqref{eq:ceiling} and the certificate bound just proved,
$\adeg(G_n)\ge\Omega\bigl(\C(G_n)^{1/4}\bigr)=\Omega(\sqrt n)$.

\smallskip
\emph{The separation.}  Fix constants $c_0>0$ and $c_1,a>0$ with
$\C(G_n)\ge c_0n^2$ and $\adeg(G_n)\le c_1\sqrt n\,\log^{a}n$ for all large $n$, as supplied by
the two displays above.  Then
\[
  \adeg(G_n)^4\ \le\ c_1^4\,n^2\log^{4a}n\ \le\ \frac{c_1^4}{c_0}\,\log^{4a}n\cdot\C(G_n),
\]
that is, $\C(G_n)=\Omt\bigl(\adeg(G_n)^4\bigr)$, proving Theorem~\ref{thm:main}.
\end{proof}

\begin{remark}[The puzzle bookkeeping]\label{rem:puzzle}
Lemma~\ref{lem:totalization} gives $\UC_1(G_n)\le k(\ell'+w)=\Ot(n)$.  Conversely, the universal
quadratic conversion $\C_0\le\UC_1^2$ \cite[Fact~1]{BGJK}, together with
$\C_0(G_n)=\Omega(n^2)$, gives $\UC_1(G_n)=\Omega(n)$.  Thus
$\UC_1(G_n)=\Tht(n)$, so the verifier bound and Lemma~\ref{lem:totalization} yield
\[
  \adeg(G_n)\le\Ot(\sqrt{\UC_1(G_n)}),
\]
which is \eqref{eq:star}.  This supplies the condition that was missing from the
exponent-$2$ puzzle solution.
\end{remark}

\subsection*{Acknowledgments}
This paper was developed with substantial assistance from ChatGPT (GPT-5.5 High),
which identified the approximate-polynomial zero-test
that completes a previously considered proof strategy.
The conversation transcript is available at \cite{Chat}.
AI assistance was used for writing, editing and polishing the paper.
The author verified the correctness and originality of all content including references.

This work was supported by the Latvian Quantum Initiative under the
European Union Recovery and Resilience Facility,
project No. 2.3.1.1.i.0/1/22/I/CFLA/001.

\appendix

\section{Proof of the surviving-pairs lemma}\label{app:sparsity}

\begin{proof}[Proof of Lemma~\ref{lem:sparsity}]
We show that uniformly random maps $r_1,\dots,r_\ell:[n]\times[n]\to[n]$ (each value chosen
independently and uniformly from $[n]$) satisfy the bound for all $S$ simultaneously with positive
probability; a fixed good choice then exists.  Fix $S$ with $|S|=m\le M_0$.  For
$(i,j)\in S\times S$ let $Y_{i,j}:=\prod_{k=1}^\ell\one[r_k(i,j)\in S]$, so that
$\Pr[Y_{i,j}=1]=(m/n)^\ell$; the variables $Y_{i,j}$ are mutually independent, as they depend on
the values of the $r_k$ at distinct points $(i,j)$.  Put $Y^S:=\sum_{(i,j)\in S\times S}Y_{i,j}$,
so
\[
  \E[Y^S]\ =\ m^2\Bigl(\tfrac mn\Bigr)^{\!\ell}\ \le\
  n^{\frac{2\ell+2}{\ell+2}}\cdot n^{-\frac{\ell}{\ell+2}}\ =\ n .
\]
By the multiplicative Chernoff bound $\Pr[X\ge(1+\eps)\mu]\le e^{-\eps^2\mu/(2+\eps)}$,
valid for a sum $X$ of independent $\bn$-valued variables with $\E[X]\le\mu$, applied with
$\mu:=n$ and $\eps:=\ell-1$,
\[
  \Pr[Y^S\ge\ell\cdot n]\ \le\
  \exp\!\Bigl(-\tfrac{(\ell-1)^2}{\ell+1}\,n\Bigr)\ \le\ e^{-(\ell-3)n},
\]
the last step using
$\tfrac{(\ell-1)^2}{\ell+1}=(\ell-3)+\tfrac{4}{\ell+1}\ge\ell-3$.  A union bound over
the at most $2^n$ choices of $S$ gives, for $\ell>4$,
\[
  \Pr\!\left[\exists S\subseteq[n],\ |S|\le M_0:\ Y^S\ge\ell n\right]
  \ \le\ 2^n e^{-(\ell-3)n}\ \le\ e^{-(\ell-4)n}\ <\ 1 .
\]
Thus some fixed maps satisfy $Y^S\le\ell n=L$ for every $S$ with
$|S|\le M_0$ simultaneously.
\end{proof}

\section{Proof of the black-box totalization}\label{app:totalization}

\begin{proof}[Proof of Lemma~\ref{lem:totalization}]
Construct $g$ from the $k$ base strings and $H$ cells exactly as described in the lemma.
Consider the $0$-input with all $k$ base strings equal to $z$.  A partial assignment of size
less than $\min\{H,\Cb{0}(F,z),\Cb{1}(F,z)\}$ leaves some cell $c\in\bn^k$ untouched.
For each $t$, its restriction to $u^{(t)}$ has size below both co-certificate complexities and
therefore has a completion with value $c_t$.  Completeness of $\varphi_{c_t}$ supplies a
descriptor for that completion; placing these descriptors in untouched cell $c$ produces a
$1$-completion.  This proves the lower bound on $\C_0(g)$.

For a $1$-input, read the $k\ell'$ bits of the uniquely addressed cell and the at most $kw$
base-input bits selected by its $k$ accepted descriptors.  Persistence keeps those same
descriptors accepted under every completion, so these bits form a $1$-certificate.
The resulting family is unambiguous: soundness fixes the address to the unique string
$(F(u^{(1)}),\dots,F(u^{(k)}))$, the cell contents then fix each set
$I_{s_t}(y_s^{(t)})$, and the base input fixes the values read on those sets.  Hence exactly one
term accepts each $1$-input and $\UC_1(g)\le k(\ell'+w)$.

It remains to prove the degree bound.  For each cell $c\in\bn^k$, let $g_c$ indicate that all
its descriptors are accepted with asserted output string $c$.  Thus $g_c$ is the
\textsc{And} of the $k$ predicates
$\varphi_{c_t}(u^{(t)},y_c^{(t)})$.  Each conjunct has approximate degree at most $d$, and
an \textsc{And} of $k=\log H=O(\log m)$ of them has approximate degree $\Ot(d)$ by
Lemma~\ref{lem:robust}(2).  By Lemma~\ref{lem:robust}(1), boost each approximant to error
$1/(3H)$ and range $[0,1]$, at an additional degree factor
$O(\log H)=O(\log m)$.  Finally $g=\bigvee_c g_c$, and \emph{at most one} exact indicator
$g_c$ is $1$ on any input: if both $g_c$ and $g_{c'}$ accept, soundness of the stored
descriptors forces the same output string to equal both $c$ and $c'$.  Hence the sum of the
approximants for the $g_c$ approximates $g$ with error
$H\cdot\tfrac1{3H}=\tfrac13$, without an additional \textsc{Or} penalty.  Therefore
$\adeg(g)\le\Ot(d)$.
\end{proof}


\end{document}